\documentclass[letterpaper,11pt]{article}
\usepackage{amsmath, amssymb}
\usepackage{color}
\usepackage{fullpage}
\usepackage[numbers]{natbib}
\usepackage[noend]{algorithmic}
\usepackage{tikz}
\usepackage{enumerate}
\usepackage{thm-restate}
\usepackage{graphicx}
\usepackage{amsopn}
\usepackage{caption}
\usepackage{hyperref}
\usepackage{subcaption}
\usepackage{zerosum,csquotes}
\usepackage{enumitem,linegoal}
\usepackage[linesnumbered,ruled,vlined]{algorithm2e}
\usepackage{comment}	
\usepackage{ctable}					

\usepackage{mathtools}
\usepackage[flushleft]{threeparttable}
\usepackage[normalem]{ulem}

\usepackage{footnote}
\makesavenoteenv{tabular}
\makesavenoteenv{table}

\usepackage[margin=1in]{geometry}

 \newtheorem{theorem}{Theorem}[section]
 \newtheorem{lemma}[theorem]{Lemma}
 \newtheorem{observation2}[theorem]{Observation}

 \newtheorem{definition}[theorem]{Definition}

\makeatletter
\def\GrabProofArgument[#1]{ #1: \egroup\ignorespaces}
\def\proof{\noindent\textbf\bgroup Proof%
	\@ifnextchar[{\GrabProofArgument}{. \egroup\ignorespaces}}

\makeatother

\DeclareMathOperator{\ev}{EVAL}
\DeclareMathOperator{\ct}{CUT}
\newcommand{\citeboth}[1]{\citeauthor{#1} \cite{#1}}

\newcommand*\samethanks[1][\value{footnote}]{\footnotemark[#1]}

\newcounter{proccnt}

\newcommand{\konote}[1]{}

\title{On the Complexity of Chore Division}

\author{
	Alireza Farhadi \thanks{University of Maryland. Email: \texttt{\{farhadi, hajiagha\}@cs.umd.edu}.}
	\thanks{Supported in part by NSF CAREER award CCF-1053605,  NSF BIGDATA grant IIS-1546108, NSF AF:Medium grant CCF-1161365, DARPA GRAPHS/AFOSR grant FA9550-12-1-0423, and another DARPA SIMPLEX grant.}
	\and MohammadTaghi Hajiaghayi \samethanks[1] \samethanks[2]
}

\begin{document}
	\newcommand{\ignore}[1]{}
\renewcommand{\theenumi}{(\roman{enumi}).}
\renewcommand{\labelenumi}{\theenumi}
\sloppy

%
%


\maketitle

\thispagestyle{empty}

\begin{abstract}
We study the proportional chore division problem where a protocol wants to divide an undesirable object, called chore, among $n$ different players. The goal is to find an allocation such that the cost of the chore assigned to each player be at most $1/n$ of the total cost. This problem is the dual variant of the cake cutting problem in which we want to allocate a desirable object. \citeboth{edmonds2006cake} showed that any protocol for the proportional cake cutting must use at least $\Omega(n \log n)$ queries in the worst case, however, finding a lower bound for the proportional chore division remained an interesting open problem. We show that chore division and cake cutting problems are closely related to each other and provide an $\Omega(n \log n)$ lower bound for chore division.
\end{abstract}
\section{Introduction}

The chore division problem is the problem of fairly allocating a divisible undesirable object among $n$ players. Imagine that we want to divide up a job among some players. There are many ways to assign this job to them, especially when players have a different evaluation of cost for each part of the job. For example, one might prefer doing certain tasks and other may not be good at them. This gives rise to an important question: How can one fairly assign a task to others?

The problem of fairly allocating an undesirable object, called chore, was first introduced by \citeboth {gardner1978aha} in the 1970s. Many definitions of fairness have been proposed for the chore division. The most important ones are \textit{proportionality} and \textit{envy-freeness}. An allocation of the chore is \textit{proportional} if everyone receives at most $1/n$ of the chore in his perspective. The other definition of fairness is envy-freeness. An allocation is \textit{envy-free} if each player receives a part that he thinks is the smallest. 

The chore division is the dual problem of the well-studied cake cutting problem in which we want to fairly allocate a divisible good among players. This problem was introduced in the 1940s, and popularized by \citeboth{robertson1998cake}. The same criteria of fairness can also be defined for this problem. For the case of two players, the simple \textit{cut-and-choose} protocol provides both proportionality and envy-freeness. In this protocol, one player cuts the cake into two equally preferred pieces, and the other chooses the best piece.

Despite the simple algorithm for the two-player proportional allocation, finding a fair allocation for more players is more challenging. In 1948, \citeboth{steinhaus1948} proposed a proportional protocol for three players. Later, Banach, Knaster, and Steinhaus proposed an $O(n^2)$ protocol inspired by the cut-and-choose protocol for proportional allocation among $n$ players. \citeboth{even1984note} improved this result by providing an $O(n \log n)$ divide and conquer protocol. Also, they showed that no protocol can proportionally allocate the cake using less than $n$ cuts; however this lower bound was not tight. The main difficulty of obtaining any lower bound for the cake cutting problem was the lack of any formal way to represent protocols. Finally, \citeboth{robertson1998cake} gave a formal model for cake cutting protocols. Their model covers almost all discrete cake cutting protocols. Later, \citeboth{edmonds2006cake} provided an $\Omega( n \log n)$ lower bound for the proportional cake cutting. Their result shows that the proportional protocol by Even and Paz is asymptotically tight.

However, finding an envy-free allocation seems to be much harder. For a long time, the only known discrete and bounded envy-free protocols were for $n \le 3$. Every other protocol required an unbounded number of queries \cite{brams1995envy,pikhurko2000envy}. \citeboth{alijani2017envy} gave a bounded protocol for the envy-free allocation under different assumptions on the valuation functions. In a recent work, \citeboth{procaccia2009thou} proved an $\Omega(n^2)$ lower bound for envy-free allocation which shows that finding an envy-free allocation is truly harder than proportional allocation. In a recent breakthrough, Aziz and Mackenzie \cite{aziz4,azizn} provided the first discrete and bounded protocol for envy-free allocation. Their protocol requires $n^{n^{n^{n^{n^n}}}}$ queries in the worst case.

Despite all the studies in the cake cutting, the results known for the chore division are very limited. The same divide and conquer algorithm by Even and Paz, finds a proportional allocation using $O( n \log n)$ queries. However, no lower bound was known for this problem. For the envy-free chore division, \citeboth{peterson2009n} gave an envy-free protocol for $n$ players, although their protocol is unbounded. Another protocol by \citeboth{peterson2002four} finds an envy-free allocation for 4 players using \textit{moving-knife} procedure. However, the moving-knife procedure is not discrete and could not be captured using any discrete protocol. Only recently, a protocol by \citeboth{dehghani2018envy} provides the first discrete and bounded protocol for the envy-free chore division. 

In this paper, we give an $\Omega( n \log n)$  lower bound for the proportional chore division. Our method shows a close relation between chore division and cake cutting. We introduce a subproblem similar to \textit{thin-rich game} introduced in \cite{edmonds2006cake}, and show that solving both proportional cake cutting and proportional chore division requires solving this problem, and solving this problem is hard. Our method can also be seen roughly as a reduction from proportional chore division to proportional cake cutting. We introduce the notion of \textit{dual} of a valuation function, and we show that how we can use dual functions to reduce some problems in chore division to similar problems in cake cutting. Since envy-freeness implies proportionality, our result also shows that any envy-free chore division protocol requires at least $\Omega( n \log n)$ queries.


\section{Preliminaries}
In chore division (resp. cake cutting) problem, we are asked to partition a divisible undesirable (resp. desirable) object, usually modeled by the interval $[0,1]$, among $n$ players. Let $N = \{1,2,\ldots,n\}$ be the set of players. Each player $i$ has a valuation function $v_i$ that indicates, given a subinterval $I\subseteq [0,1]$, the cost (resp. profit) of that interval for the player $i$.
For an interval $[x,y]$, we use $v_i(x,y)$ to denote the player $i$'s valuation for this interval. We assume that valuation functions are \textit{non-negative}, \textit{additive}, \textit{divisible} and \textit{normalized}, in other words, for each player $i$, his valuation function $v_i$ satisfies the following properties:
\begin{itemize}
\item \textit{Non-negative}: $v_i(I) \ge 0$ for every subinterval $I$ in $[0,1]$.
\item \textit{Additive}: $v_i(I_1 \cup I_2) = v_i(I_1) + v_i(I_2)$ for all disjoint intervals $I_1$ and $I_2$.
\item \textit{Divisible}: for every interval $I$ and $0 \le \lambda \le 1$, there exists an interval $I' \subseteq I$ such that $v_i(I')= \lambda v_i(I)$.
\item \textit{Normalized}: $v_i(0,1)=1$.
\end{itemize}

For an interval $I=[x,y]$, we denote $Left(I)=x$ and $Right(I)=y$. Also, we use $|I|=y-x$ to denote the width of $I$. We say that an interval $I$ is non-empty if $|I|>0$.

 We say that $P$ is a \textit{piece} of the chore if it is union of finite disjoint intervals, i.e., $P=\cup_{i=1}^{k} I_i$. For a piece $P$, we use $|P|$ to denote its width which is
$$
|P| = \sum_{i=1}^{k} |I_i| = \sum_{i=1}^{k} Right(I_i)-Left(I_i) \,.
$$
Similarly, we use $v(P)$ to denote the value of the function $v$ for $P$. It follows from additivity of valuation functions that
$$
v(P) = \sum_{i=1}^{k} v(I_i) \,.
$$
Also, we use $D_v(P)=v(P)/|P|$ to denote the density of $P$. 

The complexity of a protocol is the number of queries it makes. We use the standard Robertson and Webb query model which allows two types of queries on a valuation function $v$.
\begin{itemize}
\item ${\ev}_{v}(x,y) : $ returns $v(x,y)$.
\item ${\ct}_{v} (x,r) : $ returns $y \in [0,1]$ such that $v(x,y)=r$ or declares that answer does not exist.
\end{itemize}

An \textit{allocation} $X=(X_1,X_2,\cdots,X_n)$ is a partitioning of chore into $n$ parts  $X_1,X_2,\cdots,X_n$ such that each player $i$ receives $X_i$. We say that an allocation $X=(X_1,X_2,\cdots,X_n)$ is proportional if $v_i(X_i) \le 1/n$ for every player $i$.
\section{Lower Bound on Chore Division}
In this section we provide an $\Omega(n \log n)$ lower bound for the proportional chore division. In the cake cutting, \citeboth{edmonds2006cake} presents an $\Omega(n \log n)$ lower bound by showing that finding a dense part for an arbitrary valuation function is hard and a protocol must use at least $\Omega( \log n)$ queries. Later, they show that any proportional protocol for cake cutting finds a dense part for at least $\Omega(n)$ of $n$ arbitrary valuation functions.


For the chore division problem, we consider a special type of valuation functions in which density of each piece of the chore is at least $1/2$. Then, we represent a mapping from these valuation functions to low-density valuation functions, and show that finding any dense piece in low-density valuation functions requires $\Omega( \log n)$ number of queries. Finally, we provide a lower-bound for the proportional chore division by showing that using any protocol for this problem, we can find a dense piece for at least $\Omega(n)$ of $n$ arbitrary low-density valuation functions.

\begin{definition}
Given values $\alpha$ and $\beta$ such that $0 \le \alpha \le 1 \le \beta$, we say that a valuation function $v$ is \textit{$(\alpha,\beta)$-dense} if $\alpha \le D_v(I) \le \beta$ for every non-empty subinterval $I$ in $[0,1]$.

Moreover, a valuation function $v$ is positive if $D_v(I) >0$ for every non-empty subinterval $I$ in $[0,1]$.
\end{definition}
An example of $(\alpha,\beta)$-dense valuation functions is uniform functions. Since density of every interval in a uniform valuation function is $1$, these valuation functions are $(1,1)$-dense. 
 
For an arbitrary positive valuation function $v$, we define its dual function and show that every query on the dual function can be answered using $O(1)$ queries on the $v$. Later, we show that the dual function is an appropriate mapping from high-density functions to low-density functions. 
\begin{definition}
For a positive valuation function $v$, we use $v^{*}$ to denote its dual function and define it as follows.
$$
v^{*}(x,y) = {\ct}_{v}(0,y) - {\ct}_{v}(0,x)
$$
for every subinterval $[x,y]$ in $[0,1]$
\end{definition}
Note that in a positive valuation function, every ${\ct}_v(x,y)$ query has a unique answer, therefore the function above is well-defined for positive functions.
\begin{lemma}
\label{O1}
For a positive valuation function $v$ and its dual function $v^{*}$ the following holds.
\begin{itemize}
\item ${\ev}_{v^{*}}(x,y) = {\ct}_{v}(0,y) - {\ct}_{v}(0,x)$
\item ${\ct}_{v^{*}}(x,r) = {\ev}_{v}(0,{\ct}_{v}(0,x)+r)$
\end{itemize}
So we can answer each query on $v^{*}$ using $O(1)$ queries on $v$.
\end{lemma}
\begin{proof}
Based on definition of the dual function, we have:
$$
{\ev}_{v^{*}}(x,y)= v^{*}(x,y) = {\ct}_{v}(0,y) - {\ct}_{v}(0,x) \,.
$$
For the cut query, it should return a $y$ such that $v^*(x,y)=r$. We complete the proof by showing that $v^*(x, {\ev}_{v}(0,{\ct}_{v}(0,x)+r))=r$. By the definition of dual function, we have
\begin{align*}
&v^{*}(x,{\ev}_{v}(0,{\ct}_{v}(0,x)+r)) \\
&= {\ct}_{v}(0,{\ev}_{v}(0,{\ct}_{v}(0,x)+r)) - {\ct}_{v}(0,x) \,.
\end{align*}
Note that for any positive valuation function $v$, and $0 \le x \le 1$, we have ${\ct}_v(0,{\ev}_v(0,x))=x$. Therefore,
\begin{align*}
&v^{*}(x,{\ev}_{v}(0,{\ct}_{v}(0,x)+r)) \\
&= {\ct}_{v}(0,{\ev}_{v}(0,{\ct}_{v}(0,x)+r)) - {\ct}_{v}(0,x) \\
&= {\ct}_{v}(0,x)+r - {\ct}_{v}(0,x) = r \,.
\end{align*}
\end{proof}\\
In the following observation, we show that for a valuation function $v$, dual function of ${v^{*}}$ is $v$.
\begin{observation2}
\label{dualdual}
Let $v$ be a positive valuation function and $v^*$ be its dual function, then dual of $v^*$ is $v$. 
\end{observation2}
\begin{proof}
Let function $u$ be the dual of $v^*$, then the valuation of $u$ for an interval $[x,y]$ is
$$
u(x,y) = {\ct}_{v^*}(0,y) - {\ct}_{v^*}(0,x) \,.
$$
Since $v^*$ is the dual of $v$, by Lemma \ref{O1} we have,
\begin{align*}
&u(x,y) = {\ct}_{v^*}(0,y) - {\ct}_{v^*}(0,x) \\
&= {\ev}_v(0,{\ct}_v(0,0)+y) - {\ev}_v(0,{\ct}_v(0,0)+x) \\
&= {\ev}_v(0,y) - {\ev}_v(0,x) = v(x,y) \,. 
\end{align*}
Therefore, valuation of $u$ for any interval is the same as valuation of $v$; hence, $u=v$.
\end{proof}

We introduce high-density and low-density pieces. \citeboth{edmonds2006cake} showed that a protocol must use at least $\Omega ( \log n)$ queries in order to find a high-density piece for an arbitrary valuation function. We expand their result by showing that finding a high-density piece for a positive $(0,2)$-dense valuation function is still hard.
\begin{definition}
A piece $X$ is heavy with respect to valuation function $v$ if its width is at most $1/n$ and the valuation of $v$ on this piece be at least $1/2n$, i.e., $|X| \le 1/n$ and $v(X) \ge 1/2n$.

Similarly, a piece is light with respect to $v$ if $|X| \ge 1/2n$ and $v(X) \le 1/n$.
\end{definition}
Note that heavy and light pieces are not exclusive, and a piece can be both heavy and light.

\begin{restatable}{theorem}{findheavy}
\label{findheavy}
Any protocol that finds a heavy piece for an arbitrary positive $(0,2)$-dense valuation function requires $\Omega(\log n)$ queries in the worst case.
\end{restatable}
This theorem is our main tool to obtain a lower-bound for proportional chore division. First we show how we can use this theorem to prove the $\Omega( n \log n)$ lower bound for chore division, and then in the next section we provide a proof for this theorem.

We show that any protocol for this chore division problem requires $\Omega ( n \log n)$ queries even if all the players' valuation functions are $(1/2,\infty)$-dense. Specifically, we show that given $n$ arbitrary positive $(0,2)$-dense valuation functions, one can use their dual functions and any proportional chore division protocol to find a heavy piece for at least $\Omega(n)$ of them. First we show that if a valuation function $v$ is positive and $(0,2)$-dense, then its dual is $(1/2,\infty)$-dense.

\begin{lemma}
\label{dual}
The dual of a positive $(0,2)$-dense valuation function is positive and $(1/2,\infty)$-dense.
\end{lemma}
\begin{proof}
Suppose that $v$ is a positive $(0,2)$-dense valuation function. Let $[x,y]$ be a non-empty subinterval in $[0,1]$. We show that the density of the interval $[x,y]$ is at least $1/2$ with respect to $v^{*}$. For an interval $[x,y]$ we have $v^*(x,y)= {\ct}_v(0,y)-{\ct}_v(0,x)$. Therefore, we have
$$
D_{v^*} (x,y) = \dfrac{{\ct}_v(0,y)-{\ct}_v(0,x)}{y-x} \,.
$$
Note that $x={\ev}_v(0, {\ct}_v(0,x))$, and $y={\ev}_v(0, {\ct}_v(0,y))$. Therefore, by setting $l= {\ct}_v(0,x)$, and $r = {\ct}_v(0,y)$, we have $x={\ev}_v(0,l)$ and $y={\ev}_v(0,r)$. Therefore,
\begin{align*}
D_{v^*} (x,y) &= \dfrac{r-l}{{\ev}_v(0,r)-{\ev}_v(0,l)}\\
&= \dfrac{r-l}{{\ev}_v(l,r)} \\
&= \dfrac{1}{D_v(l,r)} \,.
\end{align*}
Since $v$ is positive $(0,2)$-dense, we have $0 < D_v(l,r) \le 2$, therefore,
$$
D_{v^*} (x,y) = \dfrac{1}{D_v(l,r)} \ge \dfrac{1}{2} \,.
$$
\end{proof}

Now, we show that if $n$ players all have a $(1/2,\infty)$-dense valuation functions, then in any proportional allocation of chore to these players, at least $n/3$ of allocated pieces are light.
\begin{lemma}
\label{light}
Given $n$ players with $(1/2,\infty)$-dense valuation functions $u_1,\cdots,u_n$, let $X_1,X_2,\cdots,X_n$ be any proportional allocation of the chore to the players such that $X_i$ is allocated to the player $i$, then at least $n/3$ of the allocated pieces are light for their owners.
\end{lemma}
\begin{proof}
For each player $i$, we use $w_i$ to denote the width of the piece allocated to player $i$, i.e., $w_i = |X_i|$. Therefore, $\sum_{i=1}^{n} w_i =1$. Since the allocation is proportional we have $u_i(X_i) \le 1/n$ for every $u_i$. Also, Since the valuation functions are $(1/2,\infty)$-dense, we have $u_i(X_i) \ge w_i/2$ for every player $i$, therefore $w_i \le 2 u_i(X_i) \le 2/n$. Now assume that $t$ is the number of pieces with the width less than $1/2n$. Since the width of every other piece is at most $2/n$, the following holds.
$$
\dfrac{t}{2n}+ \dfrac{2(n-t)} {n} \ge 1 \Rightarrow t \le \dfrac{2n}{3} \,.
$$
Therefore at least $n-t \ge n/3$ of the $w_i$ are at least $1/2n$, and the width of at least $n/3$ of the $X_i$ are at least $1/2n$. Note that because of the proportionality, the value of each of these pieces is at most $1/n$ for its owner. Therefore these pieces are light.
\end{proof}

Now we present a mapping from light pieces to heavy pieces in the dual of the valuation function.
\begin{definition}
For an interval $I=[a,b]$ and a positive valuation function $v$, we define the dual of this interval with respect to $v$ as $I_v^{*}=[\ev_v(0,a),\ev_v(0,b)]$.

For a part $P$ which is union of finite disjoint intervals $I_1,\cdots, I_k$, we define the dual of $P$ as $P_v^{*}= \cup_{i=1}^{k} {I_i}_v^{*}$.
\end{definition}
\begin{lemma}
\label{heavy}
Let $P$ be a light piece with respect to $v$ where $v$ is a positive valuation function, then $P_{v}^{*}$ is heavy with respect to $v^{*}$.
\end{lemma}
\begin{proof}
Suppose that $P= \cup_{i=1}^{k} I_i$, then $P_v^{*}= \cup_{i=1}^{k} {I_i}_v^{*}$. It follows that
\begin{align*}
|P_v^{*}|&=\sum_{i=1}^{k} |{I_i}_v^{*}|\\
&=\sum_{i=1}^{k} {\ev}_v(0,Right(I_i))- {\ev}_v(0,Left(I_i))\\&
 = \sum_{i=1}^{k} v(I_i) = v(P) \le \dfrac{1}{n} \,.
\end{align*}
Also, for an interval $I=[a,b]$, we have,
\begin{align*}
v^*(I_v^*) &= v^*(\ev_v(0,a),\ev_v(0,b)) \\
&= {\ct}_v(0,\ev_v(0,b)) - {\ct}_v(0,\ev_v(0,a)) 
\\&= b-a = |I| \,.
\end{align*}
Therefore,
$$
v^{*}(P_v^{*})=\sum_{i=1}^{k} v^{*}({I_i}_v^{*}) =  \sum_{i=1}^{k} |I_i| = |P| \ge \dfrac{1}{2n} \,.
$$
This completes the proof.
\end{proof}

Now we are ready to prove that complexity of any proportional chore division protocol is at least $\Omega( n \log n)$
\begin{theorem}
Any protocol for the proportional chore division makes at least $\Omega ( n \log n)$ queries in the worst case.
\end{theorem}
\begin{proof}
Suppose that the query complexity of a chore division protocol is $T(n)$.
Consider $n$ arbitrary positive $(0,2)$-dense valuation functions $v_1, v_2, \cdots, v_n$, and solve the proportional chore division problem for their dual functions $v^{*}_1, v^{*}_2, \cdots, v^{*}_n$. Let $X_1, X_2, \cdots, X_n$ be the pieces allocated to the players respectively. For every $v^*_i$, by Lemma \ref{O1}, we can answer each query on this function by making $O(1)$ queries on $v_i$. Therefore, we can find the proportional chore division for the dual valuation functions $v^{*}_1, v^{*}_2, \cdots, v^{*}_n$ using $O(T(n))$ queries.

According to the Lemma \ref{dual}, valuation functions  $v^{*}_1, v^{*}_2, \cdots, v^{*}_n$ are $(1/2,\infty)$-dense. Therefore, by Lemma \ref{light}, at least $n/3$ of the $X_1, X_2, \cdots, X_n$ are light with respect to the dual valuation functions. Let $Y_1, Y_2, \cdots, Y_n$ be the dual of pieces, where $Y_i$ is the dual of $X_i$ with respect to $v^*_i$, i.e., $Y_i = {X^*_i}_{v^*_i}$. Since the dual of $v^{*}_i$ is $v_i$, applying Lemma \ref{heavy} implies that at least $n/3$ of the dual pieces $Y_1, Y_2, \cdots, Y_n$ are heavy for the valuation functions $v_1, v_2, \cdots, v_n$. Since the protocol makes at most $O(T(n))$ queries, the pieces returned by this protocol are at most union of $O(T(n))$ intervals, so we can calculate the dual of all these pieces using $O(T(n))$ queries. Therefore we can find heavy piece for at least $n/3$ of the valuation functions using $O(T(n))$ queries. This along with Lemma \ref{findheavy} implies that $T(n)= \Omega( n \log n)$. Note that pieces $Y_1, Y_2, \cdots, Y_n$ are not necessarily disjoint. However, we only want to find a heavy piece for $\Omega(n)$ of valuation functions, and since this is not an allocation, $Y_1, Y_2, \cdots, Y_n$ can intersect.



\end{proof} 

\section{Lower Bound on Finding a Heavy Piece}
In this section we prove Theorem \ref{findheavy} and give an $\Omega( \log n)$ lower bound for finding a heavy piece in positive $(0,2)$-dense valuation functions. We introduce special types of valuation functions, called \textit{balanced-value-trees}, and present an adversarial strategy that gives an $\Omega( \log n)$ lower bound for finding a heavy piece on balanced-value-trees. Balanced-value-trees are very similar to \textit{value trees} valuation functions introduced in \cite{edmonds2006cake}, but value trees cannot be used for our problem since these valuation functions are not necessarily $(0,2)$-dense.

Assume that $n \ge 3^{11}$ is a power of three. A balanced-value-tree is a ternary balanced tree with $n$ leaves and depth $d=\log_3 n$. Each non-leaf node $v$ in the tree has three children, we use $l(u)$, $m(u)$ and $r(u)$ to denote its left, middle, and right child. Each node in the tree corresponds to an interval in $[0,1]$. For each node $u$, we use $I_u$ to denote the interval that corresponds to $u$, and we use $V(u)$ and $D(u)$ to denote the value and the density of the interval $I_u$. Let $r$ be the root of the tree, then $|I_r|=1$ and $V(r)=1$. For every non-leaf node $u$, its children partition the interval corresponding to $u$ into three disjoint intervals with the equal width, i.e., $I_{u} = I_{l(u)} \cup I_{m(u)} \cup I_{r(u)}$ and $|I_{l(u)}| = |I_{m(u)}| = |I_{r(u)}|= |I_u|/3$. It follows that width of every leaf in the tree is $1/n$. We call a node $u$ \textit{critical} if $D(u) \times \beta > 2$ where $\beta= 2^{6/ \ln(n)}$.

We label each edge in the tree such that the value of every node $u$ be the product of the label of edges along the path from the root to $u$. For a non-leaf critical node $u$, the label of edges between this node and its children are $1/3$ . Every other non-leaf node has an edge with label $\beta/3$ called a \textit{heavy edge}, and the label of its two other edges which we call them \textit{light} are $1/2-\beta/6$. Since $n \ge 3^{11}$, we have $1/3 \le \beta/3 < 1/2$, i.e., the value of every heavy edge is between $1/3$ and $1/2$. Also, the value of every light edge is between $1/4$ and $1/3$.  We assume that the valuation of every leaf in the tree is uniform, this means that for every leaf $u$, and an interval $I \subseteq I_u$, we have $V(I)=\dfrac{V(u) |I|}{|I_u|}$. Therefore we can find the valuation of an arbitrary interval using the tree. Note that, it follows from the definition of critical nodes and our labeling that all children of a critical node are also critical. 

\begin{lemma}
In a balanced-value-tree, the value of every non-leaf node is the sum of the values of its children.
\end{lemma}
\begin{proof}
Consider a non-leaf node $u$, if node $u$ is critical then all its edges to its children are labeled $1/3$, therefore $V(l(u))+ V(m(u))+ V(r(u)) = (V(u)/3 + V(u)/3 + V(u)/3)= V(u)$. The lemma holds in this case.

Otherwise, we suppose that a node $u$ is not critical, therefore $V(l(u))+ V(m(u))+ V(r(u)) = V(u) (\beta/3 + 1/2-\beta/6 + 1/2-\beta/6) = V(u)$. This completes the proof.
\end{proof}

For a node $u$, we use $h(u)$, $q(u)$ and $z(u)$ to denote respectively the number of heavy, light, and other edges along the path from the root to $u$. In the following lemma we show that how we can compute the density of node $u$ from $h(u)$, $q(u)$ and $z(u)$.
\begin{lemma}
\label{dentree}
In a balanced-value-tree we have $D(u)=\beta ^ {h(u)} (3/2-\beta/2) ^ {q(u)}$ for every node $u$ in the tree.
\end{lemma}
\begin{proof}
By the definition of balanced-value-trees we have
\begin{align*}
D(u) & = \dfrac{V(u)}{|I_u|} \\
& = \dfrac{(\beta/3)^{h(u)} (1/2-\beta/6)^{q(u)} (1/3)^{z(u)} }{(1/3)^{h(u)+q(u)+z(u)}} \\
& = \beta ^ {h(u)} (3/2-\beta/2) ^ {q(u)} \,.
\end{align*}
\end{proof}

Now we are ready to show that balanced-value-trees are positive and $(0,2)$-dense.
\begin{lemma}
Consider any balanced-value-tree, the valuation that this tree represents is positive and $(0,2)$-dense.
\end{lemma} 
\begin{proof}
Our goal is to show that density of every non-empty interval $I$ is at most $2$ and greater than $0$. By the definition of the balanced-value-tree, the valuation for every interval $I$ is greater than zero, so as its density.

Now it remains to show that density of every interval is at most $2$. Assume for the sake of contradiction that there is an interval with density greater than $2$, it follows that there is at least one leaf in the tree with the density greater than $2$. Let $u$ be this leaf. Consider the case that $u$ is a critical leaf, let $u_1,u_2,\cdots,u_k$ be the path from the root to $u$ where $u_1=r$ and $u_k=u$. Let $u_t$ be the first node in this sequence that is critical. $t$ is larger than one, since the root is not critical. It follows from the definition of the balanced-value-trees that $D(u)= D(u_t) \le \beta D(u_{t-1})$. Since $u_{t-1}$ is not critical, it implies that $ \beta D(u_{t-1}) \le 2$, therefore $D(u) \le 2$ which is a contradiction.

Otherwise, suppose that $u$ is not critical, therefore $D(u) < \beta D(u) \le 2$, thus we have the contradiction for both cases.  
\end{proof}

We say that a non-critical leaf $u$ in the tree is \textit{rich} if $D(u) \ge 1/2$. The following lemma shows that we can use any protocol that finds a heavy piece to find either a rich or a critical leaf in a balanced-value-tree.

\begin{lemma}
\label{reduc}
If a protocol finds a heavy piece in positive $(0,2)$-dense valuation functions with at most $T(n)$ queries, then using $O(T(n))$ queries we can find either a rich or a critical leaf in a balanced-value-tree. 
\end{lemma}
\begin{proof}
Consider the valuation function derived from the balanced-value-tree. Let piece $P$ be the output of the protocol for this valuation function. Since $P$ is heavy, we have $|P| \le 1/n$ and $V(P) \ge 1/2n$. Therefore, $D(P) \ge 1/2$. Since $P$ is the union of at most $O(T(n))$ intervals, there is an interval $I$ with the density at least $1/2$. We can find this interval with $O(T(n))$ queries. Since $|I| \le |P| \le 1/n$, the interval $I$ overlaps with at most two leaves, and one of these leaves has a density at least $1/2$ which can be found with $O(1)$ queries. The density of this leaf is at least $1/2$, thus it is either critical or rich.
\end{proof}

For the rest of the section, our goal is to show that any protocol for finding a leaf which is either rich or critical must make $\Omega( \log n)$ queries in the worst case. To this end, we first show that $h(u) = \Omega(\log n)$ for critical and rich leaves.
\begin{lemma}
\label{tooheavies}
Let $u$ be a leaf which is either rich or critical, then $h(u) > (\ln n)/6-1$.
\end{lemma} 
\begin{proof}
First, we assume that $u$ is a critical leaf, therefore $\beta D(u) >2 \Rightarrow D(u) > 2/\beta$. It implies from lemma \ref{dentree} that
$$2/\beta < D(u) = \beta ^ {h(u)} (3/2-\beta/2) ^ {q(u)} \le \beta ^ {h(u)} \,.$$
Therefore,
$$ \beta^{h(u)+1} > 2 \Rightarrow h(u) > (\log_{\beta} 2) -1 = (\ln n)/6-1 \,.
$$
Otherwise, suppose that $u$ is a rich leaf, therefore $D(u) \ge 1/2$. Since $u$ is a leaf and is not critical, we have $h(u)+q(u)=d=\log_3 n$. We complete the proof by showing that $h(u)$ must be greater than $(\ln n)/6$. For the sake of contradiction suppose that $h(u) \le (\ln n)/6$, therefore the maximum density that $u$ can have is
$$
\beta ^ { (\ln n)/6} (3/2-\beta/2)^{\log_3 n - (\ln n)/6} \,.
$$
Let denote $f(n)$ to be the function above.

It is easy to verify that function $f$ is an increasing function in $n$. Therefore,
$$
f(n) \le \lim_{x \to \infty} f(x) = 2^{3/2 - 3/\ln 3} \approx 0.426 < \dfrac{1}{2} \,.$$
which is a contradiction. Thus, for this case $h(u) > (\ln n)/6$.
\end{proof}

We now show that any protocol that finds either a critical or a rich leaf must make $\Omega( \log n)$ quries in the worst case. We give an adversarial strategy very similar to the strategy represented in \cite{edmonds2006cake} which prevents any protocol to find a critical or a rich leaf with less than $\Omega( \log n)$ queries. Consider a balanced-value-tree. At the beginning, the label of each edge is unknown to the protocol. However, after each query instead of answering the query, we reveal the label of some edges in the tree such that the answer of the query can uniquely be determined from the revealed edges. Let $u$ be a node in the tree and $u_1, u_2, \cdots, u_k$ be the path from the root to $u$ where $u_1 = r$ and $u_k = u$. We say that node $u$ is \textit{revealed} if for every $1 \le i \le k$, all labels of node $u_i$ to its children are revealed. The following lemma shows the information that a player can get from the revealed labels.  
\begin{lemma} 
(Lemma 2.2 in \cite{edmonds2006cake})
\label{information}
\begin{itemize}
\item For any revealed node $u$, the value of $V(u)$ can be determined.
\item For any revealed node $u$, values $V(0,Left(I_u))$ and $V(0,Right(I_u))$ can be computed.
\item Let $u,v$ be two revealed leaves, and $x \in I_u$ and $y \in I_v$, then values $V(0,x)$ and $V(x,y)$ can be computed.
\item Let $u$ be a revealed leaf, $x \in I_u$, $\alpha$ be a value and $v$ be the leaf that contains a point $y$ such that $V(x,y)= \alpha$, then the least common ancestor of $u$ and $v$ can be computed. 
\end{itemize}
\end{lemma}
We are now ready to provide an adversarial strategy against finding a critical or a rich leaf.
\begin{lemma}
\label{strategy}
The query complexity of any protocol that finds either a critical or a rich leaf in a balanced-value-tree is $\Omega( \log n)$.
\end{lemma}
\begin{proof}
We follow the following strategy for the first $\lfloor ((\ln n)/6 -1)/2 \rfloor$ queries. The strategy is very similar to the one in \cite{edmonds2006cake} with some minor changes, and we just highlight the main ideas in this strategy and the details of it can be found in the original paper. The strategy reveals the label of some edges after each query such that the answer of the query can be computed and keeps the following invariants. First, for each node in the tree either none or all of its edges to its children are revealed. Second, after $m$ queries, in any path from the root to a leaf at most $2m$ heavy edges are revealed. Three, all the revealed nodes form a connected component in the tree. Now we show that how we reveal the label of edges for each query.
\begin{itemize}
\item For ${\ev}(x,y)$ query, let $u_k$ be the leaf containing $x$, $u_1, u_2, \cdots, u_k$ be the path from the root to $u$, and $u_l$ be the first unrevealed node. For each $u_i$, $l \le i <k$, we reveal the label of $u_i$ to its children such that the edge between $u_i$ and $u_{i+1}$ be light, and $u_i$ has exactly two light edges and one heavy edge. Similarly, we reveal the edges in the same way for the point $y$.  Lemma \ref{information} shows that the answer of this query can be computed using these edges.
\item For ${\ct}(x,\alpha)$ query, we reveal the edges in the same way as the last case for the point $x$. Let $y$ be a point such that $V(x,y) = \alpha$. Using Lemma \ref{information}, we can find the least common ancestor of leaves containing $x$ and $y$. Let $u$ be this node. Let $u'$ be the first unrevealed node from $u$ towards the leaf containing $y$. Let $\gamma$ be the value that we must find in the subtree with the root $u'$. Recall that label of heavy edges are $1/3 \le \beta/3 < 1/2$. If $\gamma > \beta/3$, then we reveal the label of the first children to be heavy, and two others to be light. Otherwise, we reveal the first two edges to be light and the last one to be heavy. Since $\beta/3 < 1/2$, the edge between $u'$ and its child that contains $y$ is always light. We reveal the edges in the same way for the child which contains $y$ and do the same thing until we reveal the leaf containing $y$. 
\end{itemize}
We follow this strategy for the first $\lfloor ((\ln n)/6 -1)/2 \rfloor$ queries. Since $n \ge 3^{11}$, we have $((\ln n)/6 -1)/2 >0$. After these queries, at most $(\ln n)/6-1$ heavy edges are revealed in any path from the root to a leaf. By lemma \ref{tooheavies}, any critical or rich leaf must have more than $(\ln n)/6-1$ heavy edges in its path to the root.
Therefore, the protocol could not be sure about any critical or rich leaf. After these queries, the density of no interval is greater than $2$, since every critical node should have more than $(\ln n)/6-1$ heavy edges in its path to the root and no node is revealed to be critical. Therefore, this is a valid labeling of a balanced-value-tree, and no protocol can find a critical or a rich leaf against this strategy with less than $\lfloor ((\ln n)/6 -1)/2 \rfloor$ queries; hence, the query complexity of any protocol is $\Omega( \log n)$. 
\end{proof}

Now we can prove Theorem \ref{findheavy}.
\findheavy*
\begin{proof}
Let $T(n)$ be the query complexity of a protocol that finds a heavy piece. By Lemma \ref{reduc}, we can find either a critical or a rich leaf using $O(T(n))$ queries. Since query complexity of finding a critical or a rich leaf is $\Omega( \log n)$, we have $T(n) = \Omega( \log n)$.
\end{proof}
\bibliographystyle{apalike} 
\bibliography{chore}

\begin{thebibliography}{}

\bibitem[Alijani et~al., 2017]{alijani2017envy}
Alijani, R., Farhadi, M., Ghodsi, M., Seddighin, M., and Tajik, A.~S. (2017).
\newblock Envy-free mechanisms with minimum number of cuts.
\newblock In {\em Thirty-First AAAI Conference on Artificial Intelligence}.

\bibitem[Aziz and Mackenzie, 2016a]{azizn}
Aziz, H. and Mackenzie, S. (2016a).
\newblock A discrete and bounded envy-free cake cutting protocol for any number
  of agents.
\newblock In {\em Foundations of Computer Science (FOCS), 2016 IEEE 57th Annual
  Symposium on}, pages 416--427. IEEE.

\bibitem[Aziz and Mackenzie, 2016b]{aziz4}
Aziz, H. and Mackenzie, S. (2016b).
\newblock A discrete and bounded envy-free cake cutting protocol for four
  agents.
\newblock In {\em Proceedings of the forty-eighth annual ACM symposium on
  Theory of Computing}, pages 454--464. ACM.

\bibitem[Brams and Taylor, 1995]{brams1995envy}
Brams, S.~J. and Taylor, A.~D. (1995).
\newblock An envy-free cake division protocol.
\newblock {\em The American Mathematical Monthly}, 102(1):9--18.

\bibitem[Dehghani et~al., 2018]{dehghani2018envy}
Dehghani, S., Farhadi, A., HajiAghayi, M., and Yami, H. (2018).
\newblock Envy-free chore division for an arbitrary number of agents.
\newblock In {\em Proceedings of the Twenty-Ninth Annual ACM-SIAM Symposium on
  Discrete Algorithms}, pages 2564--2583. SIAM.

\bibitem[Edmonds and Pruhs, 2006]{edmonds2006cake}
Edmonds, J. and Pruhs, K. (2006).
\newblock Cake cutting really is not a piece of cake.
\newblock In {\em Proceedings of the seventeenth annual ACM-SIAM symposium on
  Discrete algorithm}, pages 271--278. Society for Industrial and Applied
  Mathematics.

\bibitem[Even and Paz, 1984]{even1984note}
Even, S. and Paz, A. (1984).
\newblock A note on cake cutting.
\newblock {\em Discrete Applied Mathematics}, 7(3):285--296.

\bibitem[Gardner, 1978]{gardner1978aha}
Gardner, M. (1978).
\newblock {\em Aha! Aha! insight}, volume~1.
\newblock Scientific American.

\bibitem[Peterson and Su, 2002]{peterson2002four}
Peterson, E. and Su, F.~E. (2002).
\newblock Four-person envy-free chore division.
\newblock {\em Mathematics Magazine}, 75(2):117--122.

\bibitem[Peterson and Su, 2009]{peterson2009n}
Peterson, E. and Su, F.~E. (2009).
\newblock N-person envy-free chore division.
\newblock {\em arXiv preprint arXiv:0909.0303}.

\bibitem[Pikhurko, 2000]{pikhurko2000envy}
Pikhurko, O. (2000).
\newblock On envy-free cake division.
\newblock {\em The American Mathematical Monthly}, 107(8):736--738.

\bibitem[Procaccia, 2009]{procaccia2009thou}
Procaccia, A.~D. (2009).
\newblock Thou shalt covet thy neighbor's cake.
\newblock In {\em IJCAI}, volume~9, pages 239--244.

\bibitem[Robertson and Webb, 1998]{robertson1998cake}
Robertson, J. and Webb, W. (1998).
\newblock Cake-cutting algorithms: Be fair if you can.

\bibitem[Steinhaus, 1948]{steinhaus1948}
Steinhaus, H. (1948).
\newblock The problem of fair division.
\newblock {\em Econometrica}, 16:101--104.

\end{thebibliography}
\newpage

\end{document}